\documentclass[a4paper,11pt]{amsart}

\usepackage{pgf,tikz}
\usetikzlibrary{arrows}
\usepackage[normalem]{ulem}
\usepackage{breqn}
\usepackage[norelsize, linesnumbered, ruled, vlined]{algorithm2e}

\numberwithin{equation}{section}

\newcommand{\CC}{\mathbb{C}}

\newcommand{\PP}{\mathbb{P}}

\newcommand{\RR}{\mathbb{R}}
\newcommand{\ZZ}{\mathbb{Z}}

\newcommand{\cal}{\mathcal}

\def\cH{{\cal H}}

\def\cR{{\cal R}}

\def\cU{{\cal U}}

\def\cW{{\cal W}}

\newcommand{\seq}[2]{ {#1}_1 ,{#1}_2 , \cdots, {#1}_{#2} }

 \DeclareMathOperator{\image}{im}

\DeclareMathOperator{\Gr}{Gr} 
 
 \DeclareMathOperator{\GL}{GL}

\theoremstyle{plain}
\newtheorem{prop}{Proposition}[section]
\newtheorem{theo}[prop]{Theorem}
\newtheorem{lemm}[prop]{Lemma}

\newtheorem*{ques}{Question}
\newtheorem{rema}[prop]{Remark}

\theoremstyle{definition}

\newtheorem{exam}[prop]{Example}

\def\sta{^{\ast}}

\def\pri{^{\prime}}
\def\sta{^*}

\def\lra{\longrightarrow}

\def\ts{\otimes}
\newcommand{\bra}[1]{\langle{#1}|}
\newcommand{\ket}[1]{|{#1}\rangle}

\def\xy{\ket{x,y}}
\def\bxy{\ket{\bar{x},y}}

\definecolor{ffqqqq}{rgb}{1,0,0}
\definecolor{uuuuuu}{rgb}{0.27,0.27,0.27}
\definecolor{xdxdff}{rgb}{0.49,0.49,1}
\definecolor{qqqqff}{rgb}{0,0,1}
\definecolor{cqcqcq}{rgb}{0.75,0.75,0.75}

\title[Upper bounds for the number of product vectors]{The number of product vectors and their partial conjugates in a pair of spaces}

\author{Joohan Na}
\address{Department of Mathematics, Seoul National University, Seoul 151-747, Korea}
\email{jhna@snu.ac.kr}

\begin{document}

\begin{abstract}
Let $D$ and $E$ be subspaces of the tensor product of the finite-dimensional Hilbert spaces $\CC^m \ts \CC^n$. We show that the number of product vectors in $D$ with their partial conjugates in $E$ is uniformly bounded depending only on $m$ and $n$ whenever it is finite. We also give an upper bound in qubit-qunit case which we expect to be sharp.
\end{abstract}

\maketitle

\section{Introduction}

The notion of quantum entanglement is now a hot research area in fundamental quantum physics and quantum information. Although it is a natural phenomenon of quantum physics, there is no classical counterpart.

A \emph{state} is a density operator on a Hilbert space, which is a positive semi-definite Hermitian operator on the Hilbert space of trace one. A multipartite quantum system can be represented by tensor products of Hilbert spaces, each of which is considered as a single quantum system.

An entangled state appears under a multipartite quantum system. For brevity, we focus only on a state on a finite dimensional bipartite quantum system $\cH_A \ts \cH_B$ in what follows. The simplest example of a state on a bipartite quantum system $\cH_A \ts \cH_B$ is a \emph{product state}, which is a projection operator onto a product vector in the Hilbert space $\cH_A \ts \cH_B$. A \emph{separable state} $\rho$ is defined as follows:
$$\rho = \sum_i p_i \ket{e_i , f_i}\bra{e_i , f_i} ,$$
where $\sum_i p_i =1, p_i \geq 0$ and $\ket{e_i , f_i}\bra{e_i , f_i}$ are product states onto product vectors $\ket{e_i , f_i} = \ket{e_i} \ts \ket{f_i} \in \cH_A \ts \cH_B$.

A state is called \emph{entangled} if it is not separable. In general, it is very hard to know whether given a state is separable or not. This is so called the separability problem, which is known as an NP-hard problem \cite{StrongNP}, \cite{NPHard03}.
 
Choi \cite{choi82} and Peres \cite{SepPeres} formulated the necessary condition for a state to be separable using positivity of its partial transpose: If a state $\rho$ is separable, then its partial transpose $\rho^{\Gamma}$ is positive semi-definite. We call this the PPT criterion.

The Horodecki family \cite{horo96}, \cite{horo_1st_example} proved that the PPT criterion is sufficient to verify separability only for $2 \ts 2$ and $2 \ts 3$ cases and found examples of entangled states with positive partial transpose(PPTES) in $3 \ts 3$ and $2 \ts 4$ cases based on Woronowicz's work \cite{worono76}. The first example of PPT entangled states in $3 \ts 3$ case appeared in \cite{choi82} in a slight different language.

Horodecki \cite{horo_1st_example} also gave us a necessary condition for separability of a PPT state $\rho$ in terms of the range of it and its partial transpose, called the range criterion, which is stated as follows: If a PPT state $\rho$ is separable, then there exists a collection $\{\ket{x_i , y_i}\}$ of product vectors such that $\{\ket{x_i , y_i}\}$ span the range $\cR(\rho)$ of $\rho$ and $\{\ket{\bar{x}_i , y_i}\}$ the range $\cR(\rho^{\Gamma})$ of $\rho^{\Gamma}$. Therefore, the following is a natural question : For given a PPT state $\rho$, how many different product vectors are there in $\cR(\rho)$ with their partial conjugates in $\cR(\rho^{\Gamma})$?

Let $D$ and $E$ be subspaces of $\CC^m \ts \CC^n$ with $\dim D^{\perp}=k, \dim E^{\perp}=l$. In order to answer the above questions, we want to know how many product vectors $\xy$ in $\CC^m \ts \CC^n$ such that
\begin{equation}\label{eq:star}
\xy \in D, \qquad \bxy \in E.
\end{equation}

Kiem \emph{et al.} \cite{kiem2011existence} studied the conditions under which there is a nonzero product vector $\xy$ with \eqref{eq:star}. To begin with, we mention the results and transform them into the form we want.

When $k+l < m+n-2$, there are infinitely many product vectors with \eqref{eq:star}. We explain it via the proof of Theorem 1 and Lemma 2 in \cite{kiem2011existence}. Kraus \emph{et al.} \cite{Sep2xN} proved in the case $m=2$.

On the other hand, when $k+l > m+n-2$, there is no nonzero product vector $\xy$ with \eqref{eq:star} for generic $D$ and $E$. We explain why it holds via the proof of Theorem 3 in \cite{kiem2011existence}. Therefore, it is reasonable that we focus only on the case $k+l = m+n-2$ in what follows.

In this paper, we show that the number of product vectors $\xy$ with \eqref{eq:star} is uniformly bounded depending only $m$ and $n$ for all subspaces $D$ and $E$ in $\CC^m \ts \CC^n$ whenever it is finite. 

Furthermore, when $m=2$, We give an upper bound $k^2+l^2$ for generic $D$ and $E$ which we expect to be sharp. Ha and Kye \cite{hakye12uniquedecomp}, \cite{hakye13counterexam}, \cite{kye13faces} found examples that there are exactly $2$, $5$ and $10$ product vectors for $2 \ts 2$, $2 \ts 3$ and $2 \ts 4$ cases respectively when $k=1$ or $l=1$ cases. These are strong evidences that the upper bound $k^2+l^2$ could be sharp. We also give a criterion whether given a pair of subspaces $D$ and $E$ belongs to the generic case or not.

Throughout this paper, by \emph{generic} we mean that there is an open dense subset of the product of the Grassmannians $\Gr(mn-k,mn) \times \Gr(mn-l,mn)$ in which the statement holds for all pairs $(D,E)$. It could be changed to be `almost all' or `almost surely' in the sense of \cite{ruskai09}, \cite{walgate08generic}.

Kraus \emph{et al.} \cite{horodecki00operational}, \cite{Sep2xN} reduced the separability of PPT states to solving a system of polynomial equations. We estimate the number of common roots of the system of polynomial equations using the techniques in algebraic geometry and topology. The difficulty is that the variables appearing in the polynomial equations are not independent complex variables, but some complex variables and their complex conjugates.

The paper organized as follows. In section 2, we state briefly some notions and properties we will use. In section 3, we prove that the number of product vectors with \eqref{eq:star} is uniformly bounded depending only $m$ and $n$ whenever it is finite. In section 4, we give an upper bound which we expect to be sharp when $m=2$ and examples.

\subsection*{Acknowledgements}
I'm grateful to my thesis advisor Young-Hoon Kiem for his detailed and helpful comments. I also wish to express my sincere gratitude to Seung-Hyeok Kye for his many suggestions and invaluable much advice for this subject.

\subsection*{Notations and Conventions}
We will work over $\CC$ unless otherwise stated. A Grassmannian $\Gr(k,n)$ is the set of all subspaces of $\CC^n$ of dimension $k$. It is a smooth complex projective variety of dimension $k(n-k)$. For a vector space $V$, the projective space $\PP(V)$ is the set of all lines in $V$ through the origin. We write $\PP^n$ for $\PP(\CC^{n+1})$. For a vector $\ket{v} \in V$ and $\ket{w} \in W$, we write $\ket{v, w}$ for $\ket{v} \ts \ket{w} \in V \ts W$ and call them a product vector of $\ket{v}$ and $\ket{w}$. We define $\bra{v, w} \in V\sta \ts W\sta$ in the same way. A vector always means a vector with unit norm unless it is a zero vector. So, by different vectors we mean they are different each other when we identify vectors up to a nonzero constant. 

\section{Preliminaries}

In this section, we briefly introduce some notions and theorems that we will use. See \cite{UAG05}, \cite{gelfand_disc} for more details.

\subsection{Resultants}
Let $F$ be a field. Let $f = \sum_{i=0}^m a_i x^i$ and $g = \sum_{i=0}^n b_i x^i$ be polynomials in the polynomial ring $F[x]$ of degree $m$ and $n$ respectively. Suppose $f$ and $g$ have a nontrivial common factor. Then it is easy to see that there are $u$ and $v$ in $F[x]$ such that the degree of $u$ (resp. $v$) is less than $n$ (resp. $m$) and $uf+vg=0$. It implies that the following linear map of $F$-vector spaces
$$F[x]_{n-1} \times F[x]_{m-1} \lra F[x]_{m+n-1}, \qquad \qquad (u,v) \mapsto uf+vg$$
is not surjective, where $F[x]_d$ denotes the set of all polynomials in $x$ of degree less than or equal to $d$. The matrix of this linear map with respect to the monomial bases is

$$Syl(f,g) := \begin{pmatrix}
a_m & a_{m-1} & \cdots & a_0 & 0 & \cdots & 0 \\ 0 & a_m & a_{m-1} & \cdots & a_0 & \cdots & 0 \\ \vdots & & \ddots & & & \ddots & \vdots \\ 0 & \cdots & 0 & a_m & a_{m-1} & \cdots & a_0 \\ b_n & b_{n-1} & \cdots & b_0 & 0 & \cdots & 0 \\ 0 & b_n & b_{n-1} & \cdots & b_0 & \cdots & 0 \\ \vdots & & \ddots & & & \ddots & \vdots \\ 0 & \cdots & 0 & b_n & b_{n-1} & \cdots & b_0
\end{pmatrix}_{.}$$

We call this the Sylvester matrix of $f$ and $g$. Since this matrix does not have full rank, its determinant should be zero. We define the \emph{resultant} $Res(f,g)$ to be the determinant of the Sylvester matrix $Syl(f,g)$. We thus have the following.

\begin{theo}
Let $f$ and $g$ be polynomials in $F[x]$. Then $f$ and $g$ have a nontrivial common factor if and only if the resultant $Res(f,g)=0$.
\end{theo}

We assume $Res(f,g) \neq 0$ and consider the equation $uf+vg=1$, where $u=\sum_{i=0}^{n-1}c_i x^i$, $v=\sum_{i=0}^{m-1}d_i x^i$. This is equivalent to

$$Syl(f,g) \cdot \begin{pmatrix}
c_0 \\ \vdots \\ c_{n-1} \\ d_0 \\ \vdots \\ d_{m-1}
\end{pmatrix} = \begin{pmatrix}
0 \\ 0 \\ \vdots \\ \vdots \\ 0 \\ 1
\end{pmatrix}_{.}$$

By Cramer's rule, $$c_i = \frac{C_i}{Res(f,g)} \quad \text{ and } \quad d_j = \frac{D_j}{Res(f,g)}_{,}$$
where both $C_i$ and $D_j$ are integer polynomials in the coefficients of $f$ and $g$. Substituting these into $u$ and $v$, we obtain the following.

\begin{theo}
There exist nonzero polynomials $A$ and $B$ such that $Af+Bg= Res(f,g)$.
\end{theo}

Now, let us consider the two variable case. Let $P$ and $Q$ be polynomials in $F[x,y]$ of degree $m$ and $n$ respectively. Then we can write
\begin{equation*}
\begin{split}
& P(x,y) = a_m(x)y^m + a_{m-1}(x)y^{m-1} + \cdots + a_0(x), \\ & Q(x,y) = b_n(x)y^n + b_{n-1}(x)y^{n-1} + \cdots + b_0(x),
\end{split}
\end{equation*}
where $a_i(x)$ and $b_j(x)$ are polynomials in $x$ of degree at most $m-i$ and $n-j$ respectively. If we think of these polynomials as polynomials in $y$ with coefficients in the field $F(x)$, we can define the resultant $Res_y(P,Q)$ of $P$ and $Q$ in the same way as in the one variable case. Since the coefficients of $P$ and $Q$ are polynomials in $x$, so is the resultant $Res_y(P,Q)$. We obtain the following by the same argument as in the one variable case.

\begin{theo}\label{resultant}
Two polynomials $P$ and $Q$ in $F[x,y]$ have a nontrivial common factor if and only if the resultant $Res_y(P,Q)$ is identically zero. Furthermore, there exist nonzero polynomials $R$ and $S$ in $F[x,y]$ such that $RP+SQ= Res_y(P,Q)$
\end{theo}

\subsection{Betti Numbers of Real Varieties}
For a given real variety $V$, Milnor \cite{milnor_betti} found out an upper bound for the sum of Betti numbers of $V$. He proved it using Alexander duality and Morse theory.

\begin{theo}[{\cite[Theorem 2]{milnor_betti}}]\label{beti_upper}
Let $V=\{f_1 =f_2 = \cdots = f_r =0\}$ in $\RR^n$. Then the sum of Betti numbers of $V$ is at most $d(2d-1)^{n-1}$ if each polynomials $f_i$ has degree $\leq d$.
\end{theo}

\section{General cases}

In this section, we show that the number of product vectors with \eqref{eq:star} is uniformly bounded for all subspaces $D$ and $E$ in $\CC^m \ts \CC^n$. We begin with the results in \cite{kiem2011existence} and modify them to fit our situation.

\begin{theo}[{\cite[Theorem 3]{kiem2011existence}}]\label{kiem}
Let $D$ and $E$ be subspaces of $\CC^m \otimes \CC^n$ with $\dim D^{\perp}=k, \dim E^{\perp}=l$.
\begin{enumerate}
\item  If $k+l <m+n-2$, there are infinitely many product vectors $\xy$ with \eqref{eq:star}.
\item If $k+l=m+n-2$ and $\sum_{r+s=m-1}(-1)^r{k \choose r}{l \choose s} \neq 0$, there exists a nonzero product vector $\xy$ with \eqref{eq:star}. 
\item If $k+l > m+n-2$, then there are no nonzero product vectors $\xy$ with \eqref{eq:star} for generic $D$ and $E$.
\end{enumerate}
\end{theo}

These modified results may be clear for people who are familiar with intersection theory. For the reader's convenience, we briefly comment the proof of ideas.

(1) follows from the fact that $(-\alpha+\beta)^k(\alpha+\beta)^l$ is always a positive dimensional class in the cohomology ring $H\sta(\PP^{m-1} \times \PP^{n-1}) \cong \ZZ[\alpha, \beta]/\langle \alpha^m, \beta^n \rangle$ according to the proof of Theorem 1 and Lemma 2 in \cite{kiem2011existence}. (3) follows from the fact that for a generic $D$, the set of $E$ for which there exists a nonzero product vector $\xy$ with \eqref{eq:star} is a proper subset in $\Gr(mn-l,mn)$ of real dimension $\dim_{\RR} \Gr(mn-l,mn) - 2(k+l-m-n+2)$.\\

A disadvantage of homological method to count intersection numbers is that homological calculations may not lead to an actual intersection number because of the orientation issue.

The first main result of this paper tells us that the number of product vectors is uniformly bounded. 

\begin{theo}\label{mxn}
The number of product vectors $\xy$ with \eqref{eq:star} is uniformly bounded depending only on $m$ and $n$ for all subspaces $D$ and $E$ in $\CC^m \ts \CC^n$ whenever it is finite.
\end{theo}

\begin{proof}
Let $\dim D^{\perp}=k$ and $\dim E^{\perp}=l$. If $k+l < m+n-2$, there are always infinitely many product vectors $\xy$ with \eqref{eq:star} by Theorem \ref{kiem}, (1). If $k+l > m+n-2$, then the number is obviously smaller than or equal to the case we choose a subspace of codimension $m+n-l-2$ containing $D$ instead of $D$. So, it is enough to consider the case $k+l=m+n-2$.

Let $D=\{z_{ij} \;|\; \sum_{i,j}A_{ij}^{(s)}z_{ij}=0 \textrm{ for } 1 \leq s \leq k, A_{ij}^{(s)} \in \CC \}$ and $E=\{z_{ij} \;|\; \sum_{i,j}B_{ij}^{(t)}z_{ij}=0 \textrm{ for } 1 \leq t \leq l, B_{ij}^{(t)} \in \CC \}$, where $(z_{ij})$ is the coordinate of the set $M_{m \times n}(\CC) \cong \CC^m \otimes \CC^n$ of all $m \times n$ complex matrices. Then $\xy$ is contained in $D$ if and only if $\sum_{i,j}A_{ij}^{(s)}x_iy_j=0$ for every $1 \leq s \leq k$, where $x=(x_i) \in \CC^m$ and $y=(y_j) \in \CC^n$. In the same way, $\bxy$ is contained in $E$ if and only if $\sum_{i,j}B_{ij}^{(t)}\bar{x}_iy_j=0$ for every $1 \leq t \leq l$. All these equations are listed as follows:

\begin{equation}\label{eq:system_general}
\begin{split}
y_1L_1^{(1)} + y_2L_2^{(1)}+ \cdots & + y_nL_n^{(1)} =0 \\
y_1L_1^{(2)} + y_2L_2^{(2)}+ \cdots & + y_nL_n^{(2)} =0 \\
& \vdots \qquad \qquad \qquad \\
y_1L_1^{(m+n-2)} + y_2L_2^{(m+n-2)}+ & \cdots + y_nL_n^{(m+n-2)} =0
\end{split}
\end{equation}

where, $ L_j^{(q)}(x)= \begin{cases} \sum_{i=1}^m A_{ij}^{(q)}x_i & \text{ if } 1 \leq q \leq k \\ \sum_{i=1}^m B_{ij}^{(q-k)}\bar{x}_i & \text{ if } k+1 \leq q \leq m+n-2 \end{cases} $\\

If there is a nonzero product vector, then the $(m+n-2) \times n$ matrix $\left( L_j^{(i)} \right)$ should have rank less than $n$. Hence, all the $n \times n$ minors of $\left( L_j^{(i)} \right)$ must be zero. If the number of common roots of all the minors is finite, we may take a hyperplane in $\CC^m$ not passing through any common roots. Hence, we may assume one of the variables $\seq{x}{m}$ to be $1$ after using an appropriate linear transformation by homogenity of minors. Note that each minor may be thought of as two real homogeneous polynomials of degree at most $n$. Then we have to solve the system of real equations of degree at most $n$ with $2m-2$ variables. By Theorem \ref{beti_upper}, the number of common roots is at most $n(2n-1)^{2m-2}$ when it is finite. To each common root of the minors, the system of linear equations \eqref{eq:system_general} in $y_i$ has at most $1$ solution in $\PP^{n-1}$. Therefore, the number of product vectors is at most $n(2n-1)^{2m-2}$, which is the uniform upper bound that we wanted.
\end{proof}

It is very hard to find the sharp upper bound for the number of product vectors $\xy$ with \eqref{eq:star} in general. However, if $(k,l) = (0, m+n-2)$ or $(m+n-2,0)$, then we can give a sharp upper bound by intersection theory in algebraic geometry. The standard reference is \cite{fulton_intersection}. The following result is mentioned in \cite{walgate08generic} without proof.

\begin{theo}[{\cite[Corollary 3.9]{walgate08generic}}]\label{segre}
Let $D$ be a subspace $\CC^m \ts \CC^n$ with $\dim D^{\perp} = m+n-2$. Then the number of product vectors in $D$ is at most $m+n-2 \choose m-1$ whenever it is finite. Moreover, the number of different product vectors in $D$ is exactly $m+n-2 \choose m-1$ for generic $D$.
\end{theo}

\begin{proof}
Let $\sigma : \PP^{m-1} \times \PP^{n-1} \lra \PP^{mn-1}$ be the Segre embedding, i.e., $(x, y)$ maps to $x \ts y$. Since $D$ is the intersection of $m+n-2$ hyperplanes, the number we want to know is nothing but the degree of the Segre variety $\image(\sigma)$. Note that the degree is the coefficient of $\alpha^{m-1} \beta^{n-1}$ of $(\alpha +\beta)^{m+n-2}$ in the integral cohomology ring $H\sta(\PP^{m-1} \times \PP^{n-1}) \cong \ZZ[\alpha, \beta]/\langle \alpha^m, \beta^n \rangle$, so it is $m+n-2 \choose m-1$.
\end{proof}

\section{Qubit-Qunit cases}

In this section, we give an upper bound expected to be sharp for generic subspaces $D$ and $E$ in $\CC^2 \ts \CC^n$ using other techniques. We start with a lemma.

\begin{lemm}\label{k2l2}
Let $P$ and $Q$ be polynomials in $\CC[x,y]$ of bidegree at most $(k,l)$ and $(l,k)$ respectively. If the resultant $Res_y(P,Q)$ is not identically zero, then the number of common roots of $P$ and $Q$ is finite and less than or equal to $k^2+l^2$.
\end{lemm}

\begin{proof}
Since the resultant $Res_y(P,Q)$ is not identically zero, $P$ and $Q$ have no nontrivial common factor by Theorem \ref{resultant}. Over $\CC$, this is equivalent to saying that the number of common roots in $\CC^2$ is finite. We may assume the $y$-coordinates of the common roots of $P$ and $Q$ are different from one another after using an appropriate linear transformation. Note that all the common roots $(\alpha, \beta)$ of $P$ and $Q$ satisfy the equation $Res_y(P,Q)(\alpha)=0$ by Theorem \ref{resultant}. So, the number of common roots of $P$ and $Q$ is less than or equal to the number of roots of $Res_y(P,Q)=0$.

We claim that the degree of $Res_y(P,Q)$ is less than or equal to $k^2+l^2$. Note that all the coefficients of $P$ (resp. $Q$) are of degree at most $k$ (resp. $l$) and each term of the determinant of the Sylvester matrix is exactly the product of $k$ coefficients of $P$ and $l$ coefficients of $Q$. Thus, the degree of $Res_y(P,Q)$ is at most $k^2+l^2$.
\end{proof}

The following is the second main result of this paper. We assume that $k,l$ and $n$ are nonnegative integers with $k+l=n$, $n \geq 2$. 

\begin{theo}\label{2xn}
Let $D$ and $E$ be subspaces of $\CC^2 \ts \CC^n$ with $\dim D^{\perp}=k, \dim E^{\perp}=l$. Then there exists a real open subvariety $\cU$ of the product of complex Grassmannians $\Gr(2n-k,2n) \times \Gr(2n-l,2n)$ in which the number of product vectors $\xy$ with \eqref{eq:star} is less than or equal to $k^2 + l^2$ for all $D$ and $E$. 
\end{theo}

\begin{proof}
In the same way as in the proof of Theorem \ref{mxn}, if there is a nonzero product vector, the determinant of the $n \times n$ matrix $\left( L_j^{(i)} \right)$ should be zero. In order to find roots of the equation $\det(L_j^{(i)})=0$, we may assume $x_2=1$ after a linear transformation. Note that the polynomial $P(x_1, \bar{x}_1)=\det(L_j^{(i)})$ in $x_1$ and $\bar{x}_1$ has bidegree at most $(k,l)$. Let $Q(x_1,\bar{x}_1)$ be the conjugate of the polynomial $P$, i.e., $Q(x_1,\bar{x}_1)=\overline{P(x_1, \bar{x}_1)}$. If we think of the polynomials $P$ and $Q$ as the polynomials in two independent variables $z$ and $w$, then the resultant $Res_w(P,Q)$ is a polynomial in $z$. Note that all the coefficients of $Res_w(P,Q)$ are polynomials of the coefficients of $P$ and their conjugates. We regard the projective space $\PP^{(k+1)(l+1)-1}$ as the set of all polynomials in two variables of bidegree at most $(k,l)$. If we think of $P$ as an element of the projective space $\PP^{(k+1)(l+1)-1}$, then the locus where not all the coefficients of $Res_w(P,Q)$ are zero is a real open subvariety $\cW\pri$ in the complex projective space $\PP^{(k+1)(l+1)-1}$. Since all the roots of $P(z, \bar{z})=0$ are also the common roots of $P(z,w)=Q(z,w)=0$, the number of roots of $P(x_1, \bar{x}_1)=0$ is less than or equal to $k^2+l^2$ whenever $P$ is in $\cW\pri$ by Lemma \ref{k2l2}. To each root of $P(x_1, \bar{x}_1)=0$, the system of linear equations \eqref{eq:system_general} should have rank exactly $n-1$, because otherwise there are infinitely many vectors. That \eqref{eq:system_general} has rank precisely $n-1$ is a Zariski open condition because it is equivalent to saying that at least one of the $(n-1) \times (n-1)$ minors is not zero. We let $\cW$ be the intersection of $\cW\pri$ and the locus where at least one of the $(n-1) \times (n-1)$ minors is not zero for every common root of $P$ and $Q$.

We have to show that there is a real open subvariety $\cU$ of the product of complex Grassmannians $\Gr(2n-k,2n) \times \Gr(2n-l,2n)$ such that $P$ and $Q$ have finitely many common roots for any $(D,E) \in \cU$.  We define a map 
$$\phi : \Gr(2n-k,2n) \times \Gr(2n-l,2n) \lra \PP^{(k+1)(l+1)-1}$$ 
as follows: Each pair $(D,E)$ is mapped to the polynomial $P$ defined above. To be precise, for each pair $(D,E) \in \Gr(2n-k,2n) \times \Gr(2n-l,2n)$, we can write $D=\bigcap_{i=1}^k V_i^{\perp}$ and $E=\bigcap_{j=1}^l W_j^{\perp}$ for some vectors $V_i, W_j$ in $\CC^2 \ts \CC^n$. If we consider $V_i$ and $W_j$ as $2 \times n$ matrices via $M_{2 \times n} \cong \CC^2 \ts \CC^n$, we denote by $V_i^k$ and $W_j^k$ the $k$-th row of $V_i$ and $W_j$ for $k=1,2$, respectively. If we consider the $2n \times n$ matrix
$$\left( \begin{array}{c}
\text{---------}\text{---------} \quad V_1^1 \quad \text{---------}\text{---------} \\ \text{---------}\text{---------} \quad V_1^2 \quad \text{---------}\text{---------} \\ \vdots \\ \text{---------}\text{---------} \quad V_k^1 \quad \text{---------}\text{---------} \\ \text{---------}\text{---------} \quad V_k^2 \quad \text{---------}\text{---------} \\
\text{---------}\text{---------} \quad W_1^1 \quad \text{---------}\text{---------} \\ \text{---------}\text{---------} \quad W_1^2 \quad \text{---------}\text{---------} \\ \vdots \\ \text{---------}\text{---------} \quad W_l^1 \quad \text{---------}\text{---------} \\
\text{---------}\text{---------} \quad W_l^2 \quad \text{---------}\text{---------}
\end{array} \right)_{,}$$
the coefficient of a monomial $z^{k_1}w^{l_1}$ of $P(z, w)$ is the sum of all $n \times n$ minors consisting of $k_1$ $V_i^1$'s and $l_1$ $W_j^1$'s and $n-k_1-l_1$ $V_i^2$ or $W_j^2$'s.

For another representations of $D = \bigcap_{i=1}^k (V_i\pri)^{\perp}$ and $E = \bigcap_{k=1}^l (W_j\pri)^{\perp}$, the matrix whose rows are $V_i$ (resp. $W_j$) and the matrix whose rows are $V_i\pri$ (resp. $W_j\pri$) differ by left multiplication of a matrix $M_V$ in $\GL(k)$ (resp. $M_W$ in $\GL(l)$) when we consider $V_i, V\pri_i, W_j, W\pri_j$ as vectors in $\CC^2 \ts \CC^n$, i.e.,
$$
\left( \begin{array}{c}
\text{--} V\pri_1 \text{--} \\ 
\vdots \\
\text{--} V\pri_k \text{--} \\
\end{array} \right) = M_V
\left( \begin{array}{c}
\text{--} V_1 \text{--} \\ 
\vdots \\
\text{--} V_k \text{--} \\
\end{array} \right)_{,} \quad
\left( \begin{array}{c}
\text{--} W\pri_1 \text{--} \\ 
\vdots \\
\text{--} V\pri_k \text{--} \\
\end{array} \right) = M_W
\left( \begin{array}{c}
\text{--} W_1 \text{--} \\ 
\vdots \\
\text{--} W_k \text{--} \\
\end{array} \right)_{.}
$$

Then the images of $\phi$ differ by $\det(M_V)\det(M_W)$, that is, it is unchanged as an element of the projective space. Hence, $\phi$ is well-defined and algebraic. 

Now, we claim that $\phi^{-1}(\cW)$ is nonempty : We take $D = \bigcap_{j=1}^k\{z_{1j}+ j z_{2j}=0\}$ and $E = \bigcap_{j=k+1}^n\{z_{1j}- j z_{2j}=0\}$, where $(z_{ij})$ is the coordinate of the set $M_{m \times n}(\CC) \cong \CC^m \otimes \CC^n$ of all $m \times n$ complex matrices. Then $P(z,w)=(z+1)(z+2)\cdots(z+k)(w-k-1)\cdots(w-n)$ and $Q(z,w)=(z-k-1)\cdots(z-n)(w+1)\cdots(w+k)$. By direct calculation, the resultant $Res_w(P,Q)$ of $P$ and $Q$ is $\prod_{j=1}^k \prod_{i=k+1}^n (i+j)(z-i)^k (z+j)^l$, which is not identically zero. For each root of $Res_w(P,Q)$, the matrix
$$\begin{pmatrix}
z+1 & & & & && \\ & z+2 & & & && \\ & & \ddots & & && \\ & & & z+k & && \\ &&&& \bar{z}-k-1 && \\ &&&&& \ddots & \\ &&&&&& \bar{z}-n
\end{pmatrix}
$$
has rank exactly $n-1$. Hence, there is only one solution for $y$ in $\PP^{n-1}$ to each $x$ with $P(x,\bar{x})=0$. Therefore, $P$ should be located in $\cW$. Since $\phi$ is algebraic, $\phi^{-1}(\cW)$ is a real open subvariety, which is $\cU$ that we wanted.
\end{proof}

A real open subvariety $U$ of $X$ is, by definition, a complement of a proper real subvariety, i.e, $U=X-Y$, where $Y$ is a proper real subvariety of a real variety $X$. Note that a real open subvariety $U$ is an open dense subset of $X$. Hence, we reformulate the theorem above briefly as follows.

\begin{rema}
The number of product vectors $\xy$ with \eqref{eq:star} is less than or equal to $k^2 + l^2$ for generic $D$ and $E$ with $\dim D^{\perp}=k, \dim E^{\perp}=l$.
\end{rema}

As we follow the proof of theorem above carefully, we can determine whether given a pair $(D,E)$ in the product of Grassmannians $\Gr(2n-k,2n) \times \Gr(2n-l,2n)$ lies in $\cU$ or not. We give an algorithm to determine it below.

\begin{algorithm}
\SetKwRepeat{Repeat}{repeat}{until}
\KwData{$(D,E) \in \Gr(2n-k,2n) \times \Gr(2n-l,2n)$}
\KwResult{whether $(D,E)$ lies in $\cU$ or not.}
Consider the system of equations like $\eqref{eq:system_general}$\;
Produce the polynomials $P(z,w)$ and $Q(z,w)$\;
Calculate their resultant $Res_w(P,Q)$\;
\eIf{$Res_w(P,Q)$ is identically zero}{
$(D,E) \notin \cU$}
{\ForEach{root $z$ of $Res_w(P,Q)$}
{\ForEach{$(n-1) \times (n-1)$ minor $M(x_1,\bar{x}_1)$ of the matrix $\left( L^{(i)}_j \right)$}{Calculate $M(z,\bar{z})$\;}\If{all of the $M(z,\bar{z})$ are zero}{$(D,E) \notin \cU$}
}
{$(D,E) \in \cU$}}
\caption{How to determine whether a pair $(D,E)$ lies in $\cU$ or not.}
\end{algorithm}

We apply the algorithm to an example, which appears in \cite{hakye13counterexam}.

\begin{exam}
Let $a$ and $b$ be real numbers with the relation $0 < b < 4a^3/27$. Let $D=\{(z_{ij}) \in \CC^2 \ts \CC^4 \; | \; z_{12}-z_{21}=0, z_{13}-z_{22}=0, z_{14}-z_{23}=0 \}$ and $E=\{(z_{ij}) \in \CC^2 \ts \CC^4 \; | \; b z_{11}+z_{14}-a z_{22}=0\}$. If there exists a nonzero product vector $\xy$ with \eqref{eq:star}, then the determinant of the matrix
$$L = \left(\begin{array}{cccc}
-x_2 & x_1 & 0 & 0 \\ 0 & -x_2 & x_1 & 0 \\ 0 & 0 & -x_2 & x_1 \\ b\bar{x}_1 & -a\bar{x}_2 & 0 & \bar{x}_1
\end{array}\right)$$
should be zero. The determinant of $L$ is $a x_1^2 x_2 \bar{x}_2 - b x_1^3 \bar{x}_1 - x_2^3 \bar{x}_1$. Since $(x_1, x_2)=(1,0)$ is not a root of the equation $\det L =0$, we may assume $x_2=1$. Then the determinant $\det L$ turns into $a x_1^2 -b x_1^3 \bar{x}_1 -\bar{x}_1$. We let $P(z,w)= - bz^3 w + a z^2 -w$ and $Q(z,w)= - bw^3 z + a w^2 -z$. Their resultant $Res_w(P,Q)= b^3 z^{10} + 3 b^2 z^7 + (3b - a^3) z^4 + z$. Note that the $3 \times 3$ minor of $L$ consisting of first 3 rows and first 3 columns is $-1$, that is, not zero. Hence, $L$ has rank 3 for every root of $Res_w(P,Q)$. It means $(D,E)$ lies in $\cU$. Therefore, the number of product vectors $\xy$ with \eqref{eq:star} is less than or equal to $10$.
\end{exam}

\begin{exam}
Let $D$ and $E$ be subspaces $\CC^2 \ts \CC^n$ with $\dim D^{\perp}=k, \dim E^{\perp}=l$.
\begin{enumerate}
\item $2 \otimes 3$ case : \\
The number of product vectors $\xy$ with \eqref{eq:star} is
$$\begin{cases}
\leq 3 \qquad & \textrm{whenever it is finite if $k=0$ or $l=0$} \\
\leq 5 \qquad & \textrm{for almost all $D, E$ if $(k,l)=(2,1)$ or $(1,2)$} \\
\end{cases}$$

\item $2 \otimes 4$ case : \\
The number of product vectors $\xy$ with \eqref{eq:star} is
$$\begin{cases}
\leq 4 \qquad & \textrm{whenever it is finite if $k=0$ or $l=0$} \\
\leq 10 \qquad & \textrm{for almost all $D, E$ if $(k,l)=(3,1)$ or $(1,3)$} \\
\leq 8 \qquad & \textrm{for almost all $D, E$ if $(k,l)=(2,2)$}
\end{cases}$$
\end{enumerate}
\end{exam}

Kye \cite{kye13faces} described the conditions for which the number of product vectors is $0,1,2$ and $\infty$ explicitly in $2 \ts 2$ case. Ha and Kye \cite{hakye12uniquedecomp} found out an example in $2 \ts 3$ case in which there are exactly $5$ product vectors for $(k,l)=(1,2) \textrm{ or } (2,1)$. Recently, they \cite{hakye13counterexam} also  found out examples in $2 \ts 4$ case in which there are exactly $10$ product vectors $\xy$ with \eqref{eq:star} for $(k,l)=(1,3) \textrm{ or } (3,1)$. These are strong evidences that the upper bound $k^2+l^2$ could be sharp.

By the range criterion, if the number of product vectors $\xy$ with \eqref{eq:star} is less than either the dimension of $D$ or that of $E$, then all the density operators $\rho$ satisfying $\cR(\rho) = D$ and $\cR(\rho^{\Gamma}) = E$ are not separable.

For instance, even though $k \neq 0$ and $l \neq 0$, if some monomials do not appear in the polynomial $P$ in the proof of Theorem \ref{2xn}, then the degree of resultant $Res_y(P,Q)$ could be smaller than the dimension of $D$.

\begin{exam}
Let \\
$D=\left\lbrace(z_{ij}) \in \CC^2 \ts \CC^4 \; \middle| \;
\begin{matrix}
z_{11}-z_{12}+3z_{13}-3z_{14}+2z_{21}+(1+i)z_{22}=0, \\ (-2+3i)z_{11}+3z_{14}+z_{21}+2z_{22}+(7-i)z_{23}-z_{24}=0
\end{matrix}
\right\rbrace$
\\ and \\
$E=\left\lbrace(z_{ij}) \in \CC^2 \ts \CC^4 \; \middle| \; 
\begin{matrix}
11z_{11}+3z_{12}+z_{13}-2z_{23}=0, \\
(13-39i)z_{21}-(33-9i)z_{24}=0
\end{matrix}
\right\rbrace$.

If there exists a nonzero product vector $\xy$ with \eqref{eq:star}, then the determinant of the matrix
$$L = \left(\begin{array}{cccc}
x_1+2x_2 & -x_1+(1+i)x_2 & 3x_1 & -3x_1 \\ (-2+3i)x_1+x_2 & 2x_2 & (7-i)x_2 & 3x_1-x_2 \\ 11\bar{x}_1 & 3\bar{x}_1 & \bar{x}_1-2\bar{x}_2 & 0 \\ (13-39i)\bar{x}_2 & 0 & 0 & (-33+9i)\bar{x}_2
\end{array}\right)$$
should be zero. Note that $(1,0)$ is a root of the equation $\det L=0$ and gets the rank of $L$ to be $3$. In order to find other roots of the equation $\det L=0$, we may assume $x_2=1$. Then the determinant of $L$ turns into 
$$(4630+120i)x_1^2\bar{x}_1^2 + (16+492i)x_1 - (2308+1876i)\bar{x}_1 + (284-172i).$$ 
Let
\begin{equation*}
\begin{split}
&P(z,w)=(4630+120i)z^2w^2 + (16+492i)z - (2308+1876i)w + (284-172i), \\ &Q(z,w)=(4630-120i)z^2w^2 - (2308-1876i)z + (16-492i)w + (284+172i).
\end{split}
\end{equation*}
Then the resultant $Res_w(P,Q)$ of $P$ and $Q$ is 
\begin{small}
\begin{equation*}
\begin{split}
&(84363903624000 - 135466487852800i) z^6- (20371955468800 + 36686447532800i) z^5 \\ &- 2758522374400 z^4 - (95024101577600 + 52766808889600i) z^3 \\ &+ (-2037183324800 + 10244842192000i) z^2.
\end{split}
\end{equation*}
\end{small}
Note that $x_1=0$ is not a root of the equation $\det L=0$. Hence,  it is enough that we only investigate the roots of
\begin{small}
\begin{equation*}
\begin{split}
&T(z) := \frac{Res_w(P,Q)}{z^2} = \\ &(84363903624000 - 135466487852800i) z^4 - (20371955468800 + 36686447532800i) z^3 \\ &- 2758522374400 z^2 - (95024101577600 + 52766808889600i) z \\ & + (-2037183324800 + 10244842192000i).
\end{split}
\end{equation*}
\end{small}
Consider the $3 \times 3$ minor of last $3$ rows and first $3$ columns of the matrix $L$. It is $M(x_1) := (13-39i)((-19+3i)\bar{x}_1-4)$. Since $- 4/(19+3i)$ is the root of $M$, but not a root of $T$, the rank of $L$ should be $3$ for all the roots of $T$. Hence, the number of product vectors $\xy$ with \eqref{eq:star} is not more than $1 + \deg T = 5$. Since the dimension of $D$ is $6$, there is no collection of product vectors $\xy$ with \eqref{eq:star} and spanning $D$. Therefore, by the range criterion, all the density operators $\rho$ satisfying $\cR(\rho) = D$ and $\cR(\rho^{\Gamma}) = E$ are not separable. 
\end{exam}

\section{Conclusions and Discussions}

In this paper, we show that the number of product vectors $\xy$ with \eqref{eq:star} is uniformly bounded for subspaces $D$ and $E$ in $\CC^m \ts \CC^n$.

In section 4, we give an upper bound $k^2+l^2$ for generic $D$ and $E$ with $\dim D^{\perp}=k, \dim E^{\perp}=l$ when $m=2$ and $k+l = n$. By some known examples, we expect this upper bound could be sharp. As far as we know, there is no example the number of product vectors is bigger than $k^2+l^2$.

\begin{ques}
Let $D$ and $E$ be subspaces of $\CC^2 \ts \CC^n$ with $\dim D^{\perp}=k, \dim E^{\perp}=l$. We assume $k+l \geq n$. Is the number of product vectors $\xy$ with \eqref{eq:star} always at most $k^2+l^2$ when it is finite?
\end{ques}

We consider the following: If the number of linearly independent product vectors is smaller than either the rank of $\rho$ or that of $\rho^{\Gamma}$, $\rho$ cannot be separable by the range criterion.

As its applications, we give an example that the number of different product vectors is less than the dimension of $D$ although the upper bound $k^2+l^2$ is bigger than the dimension of $D$.

\bibliographystyle{abbrv}
\bibliography{Quantum}

\end{document}